\RequirePackage{fix-cm}
\documentclass[smallcondensed]{svjour3} 

\usepackage[utf8]{inputenc} 
\usepackage[T1]{fontenc}
\usepackage{url}
\usepackage{ifthen}
\usepackage{cite}
\usepackage[cmex10]{amsmath} 
\usepackage{amssymb}
\usepackage{mathtools,blkarray}                             
\usepackage{xifthen,xparse}

\usepackage[normalem]{ulem}
\usepackage{cancel}

\usepackage{hyperref}
\usepackage{float}
\usepackage{supertabular}

\usepackage{tikz,pgfplots}
\pgfplotsset{compat=newest}
\pgfplotsset{plot coordinates/math parser=false}
\usetikzlibrary{decorations.markings,calc,positioning,matrix}

\allowdisplaybreaks

\newif\iffull
\smartqed

\newcommand{\norm}[1]{\left\| #1 \right\|}	

\newcommand{\ket}[1]{\left\lvert #1 \right\rangle}

\NewDocumentCommand\ketbra{+m+g}{%
  \IfNoValueTF{#2}
    {\left\lvert #1 \right\rangle \left\langle #1 \right\vert}
  {\left\lvert #1 \right\rangle \left\langle #2 \right\rvert}%
}
\NewDocumentCommand\braket{+m+g}{%
  \IfNoValueTF{#2}
    {\left\langle #1 \vert #1 \right\rangle}
  {\left\langle #1 \vert #2 \right\rangle}%
}

\newcommand{\llbr}{[\![}
\newcommand{\rrbr}{]\!]}

\newcommand{\syminn}[2]{\langle #1, #2 \rangle_{\text{s}}}

\newif\ifcomment
\commenttrue

\newcommand{\edit}[1]{\textcolor{black}{#1}}

\begin{document}
\title{Approximate Unitary 3-Designs from\\ Transvection Markov Chains} 


\author{%
	Xinyu Tan \and 
	Narayanan Rengaswamy \and
	Robert Calderbank%
}

\institute{X. Tan \and R. Calderbank \at
              Department of Mathematics, Duke University, Durham, North Carolina 27708, USA. 
           \and
           N. Rengaswamy \at
              Department of Electrical and Computer Engineering, The University of Arizona, Tucson, Arizona 85721, USA. 
           Most of this work was conducted when he was with the Department of Electrical and Computer Engineering, Duke University, Durham, North Carolina 27708, USA.\\
        \email{\{xinyu.tan, robert.calderbank\}@duke.edu, narayananr@arizona.edu} 
}


\maketitle

\begin{abstract}
Unitary \edit{$k$-designs} are probabilistic ensembles of unitary matrices whose first \edit{$k$} statistical moments match that of the full unitary group endowed with the Haar measure.
In prior work, we showed that the automorphism group of classical $\mathbb{Z}_4$-linear Kerdock codes maps to a unitary $2$-design, which established a new classical-quantum connection via graph states.
In this paper, we construct a Markov process that mixes this Kerdock $2$-design with symplectic transvections, and show that this process produces an $\epsilon$-approximate unitary $3$-design. 
We construct a graph whose vertices are Pauli matrices, and two vertices are connected by \edit{directed edges} if and only if they commute.
A unitary ensemble that is transitive on vertices, edges, and non-edges of this Pauli graph is an exact $3$-design, and the stationary distribution of our process possesses this property.
With respect to the symmetries of Kerdock codes, the Pauli graph has two types of edges; the Kerdock $2$-design mixes edges of the same type, and the transvections mix the types.
More precisely, on $m$ qubits, the process samples \edit{$O(\log(N^5/\epsilon))$} random transvections, where $N = 2^m$, followed by a random Kerdock $2$-design element and a random Pauli matrix.
Hence, the simplicity of the protocol might make it attractive for several applications.
From a hardware perspective, $2$-qubit transvections exactly map to the M{\o}lmer-S{\o}rensen gates that form the native $2$-qubit operations for trapped-ion quantum computers.
Thus, it might be possible to extend our work to construct an approximate $3$-design that only involves such $2$-qubit transvections.
\end{abstract}


\section{Introduction}
\label{sec:intro}

Unitary \edit{$k$-designs} are finite collections of $m$-qubit unitary matrices endowed with a probability distribution, and they satisfy a certain statistical criterion.
Set $N \coloneqq 2^m$.
Then, the first \edit{$k$} statistical moments of this finite ensemble match that of the unique rotationally-invariant Haar measure on the group $\mathbb{U}_N$ of all $N \times N$ unitary operators.
Unitary designs serve several purposes in quantum information science such as quantum data hiding~\cite{DiVincenzo-it02}, decoupling in quantum information theory~\cite{Roy-dcc09,Nakata-jmp17,Hayden-arxiv07,Szehr-njs13}, and channel fidelity estimation~\cite{Dankert-physreva09}.

\emph{Randomized benchmarking} is a protocol commonly used to estimate the quality of state preparation, measurement and operations (gates) in a quantum computer~\cite{Emerson-joptb05,Magesan-physreva12}.
The actual errors in the system could be gate- and time-dependent, so estimating the error environment completely is challenging.
Therefore, the procedure attempts to estimate the \emph{average gate fidelity} that characterizes the error environment.
The protocol works by first preparing a fixed initial state $\ket{\psi}$.
Then, for a fixed sequence length $s$, one generates $K_s$ sequences of $(s+1)$ operations each, where the first $s$ operations are chosen randomly from a unitary $2$-design and the last operation is defined to be the inverse of the composition of the first $s$ operations.
Hence, ideally, the final state should be identical to the initial state.
For each sequence, one measures the average survival probability of the state, i.e., $\text{Tr}[\Pi_{\psi} \mathcal{S}(\ketbra{\psi}) ]$, where $\Pi_{\psi}$ is the POVM element to detect $\ket{\psi}$ and $\mathcal{S}(\cdot)$ is the effective channel induced by the aforementioned sequence of (noisy) operations.
If the POVM is realized ideally, then $\Pi_{\psi} = \ketbra{\psi}$.
Then, one averages the results over the $K_s$ sequences to compute the average sequence fidelity.
This procedure is repeated for each $s$ and the results are fit to a fidelity decay function.

This scheme amounts to ``twirling'' the underlying noise channel using a unitary $2$-design in order to arrive at the depolarizing channel with the same fidelity as the original channel~\cite{Emerson-joptb05,Magesan-physreva12}.
Then, the fidelity is estimated on this single-parameter depolarizing channel in order to establish the quality of the computing environment.
Small, practically feasible unitary 2-designs make it possible to efficiently characterize the reliability of a quantum computing environment, thereby enabling the development of quantum computers.

In prior work, we showed that the unitary 2-design constructed by Cleve et al.\cite{Cleve-arxiv16} coincided with the symmetry group of the $\mathbb{Z}_4$-linear Kerdock code~\cite{Can-arxiv19}.
The Kerdock codewords appear as graph states, providing a new connection between classical and quantum information theory.
Our $2$-design is a subgroup of the Clifford group, and the corresponding group of binary symplectic matrices, $\mathfrak{P}_{m}$, is isomorphic to the projective special linear group PSL($2,2^m$).

It is well-known that there is no proper subgroup of the Clifford group that can form a unitary $3$-design~\cite{Webb-arxiv16}.
In this paper, we combine our (Kerdock) $2$-design with symplectic transvections~\cite{Salam-laa08,Koenig-jmp14,Rengaswamy-tqe20} to construct a Markov process that results in an approximate unitary $3$-design.
Hence, our work demonstrates how one can ``smoothly'' turn the Kerdock $2$-design into a $3$-design.

\textbf{Main Result (Theorem~\ref{thm:sample_transvections}):}
Random sampling of \edit{$O(\log(N^5/\epsilon))$} transvections, followed by a random element from $\mathfrak{P}_{m} \cong \text{PSL}(2,2^m)$ and a random Pauli matrix, produces an $\epsilon$-approximate unitary $3$-design.

An exact unitary $3$-design must be transitive on all Paulis, on \edit{ordered} pairs of commuting Paulis, and on \edit{ordered} pairs of anti-commuting Paulis.
The Kerdock $2$-design acts transitively on Pauli elements, but partitions the Pauli pairs into multiple orbits. 
Using a finite field representation of Paulis, and the fact that transvections generate the symplectic group, we characterize the orbits of Pauli pairs and analyze how each transvection acts on the orbits. 
Finally, we analyze the convergence rate to an $\epsilon$-approximate $3$-design using the second largest eigenvalues of the transition matrices of the ``edge''- and ``orbit''-Markov chains.  

While only few applications exist currently for unitary $3$-designs~\cite{Brandao-qic13,Kueng-arxiv16,Kueng-arxiv16b}, we think that the simplicity of our protocol makes it an attractive candidate for any such application.
Our Markov process samples Clifford transformations uniformly from cosets of the Kerdock $2$-design that are determined by products of transvections.
Since transvections form a conjugacy class inside the Clifford group (see~\eqref{eq:transvec_conjugacy}), the intermediate Cliffords in the Markov process (from $\mathfrak{P}_{m} \cong \text{PSL}(2,2^m)$) can be combined into one final Clifford (from $\mathfrak{P}_{m} \cong \text{PSL}(2,2^m)$) as stated in the result above.
We emphasize here that, while the full Clifford group forms a $3$-design~\cite{Webb-arxiv16}, transvections form a specific subset whose structure could be exploited for practical implementations.
In particular, $2$-qubit transvections exactly correspond to M{\o}lmer-S{\o}rensen gates that form the native $2$-qubit operations in a trapped-ion quantum computer~\cite{Linke-nas17,Rengaswamy-phd20}.
Therefore, it might be possible to suitably modify our approach in this paper to construct an approximate $3$-design with only $2$-qubit transvections.
This way, the design could be tailor-made for trapped-ion systems.

The rest of the paper is organized as follows. In Section~\ref{sec:pauli_clifford}, we introduce our finite field representation of the Pauli group elements, show how it expresses commutativity, and discuss other preliminaries including the Clifford group and symplectic transvections. In Section~\ref{sec:pauli_geometry}, we define a graph on Pauli matrices, where Clifford elements act as graph automorphisms, and explain how the Pauli graph works with the ideas of Pauli mixing and Pauli $2$-mixing. In Section~\ref{sec:unitary_2}, we introduce the Kerdock unitary $2$-design which is a symplectic subgroup isomorphic to the projective special linear group $\text{PSL}(2, 2^m)$. Then, using our finite field representation, we define orbit invariants based on how $\text{PSL}(2, 2^m)$ partitions the \edit{directed} edges of the Pauli graph. In Section~\ref{sec:markov}, we introduce the transvection Markov chains whose stationary distribution gives an exact unitary $3$-design. In Section~\ref{sec:convergence}, we analyze the convergence rate of the transvection Markov process and prove that it produces an $\epsilon$-approximate unitary $3$-design. In Section~\ref{sec:conclusion}, we conclude the paper.

\section{The Pauli and Clifford Groups} \label{sec:pauli_clifford}

In this section, we describe commutativity in the Pauli group by rewriting the standard symplectic inner product as a trace inner product over a finite field. This translation simplifies the description of the Clifford symmetries that generate the Kerdock unitary 2-design.

\subsection{The Finite Field $\mathbb{F}_{2^m}$}

The field representation is fundamental to our description of commutativity in the Pauli group. We obtain the finite field $\mathbb{F}_{2^m}$ from the binary field $\mathbb{F}_2$ by adjoining a root $\alpha$ of a primitive irreducible polynomial $p(x)$ of degree $m$~\cite{McEliece-1987}. Each element of $\mathbb{F}_{2^m}$ corresponds to a polynomial in $\alpha$ of degree at most $m-1$ with coefficients in $\mathbb{F}_2$. The field elements $1, \alpha, \alpha^2, \ldots, \alpha^{m-1}$ form a basis for $\mathbb{F}_{2^m}$ over $\mathbb{F}_2$ which we call the \emph{primal basis}. 
The corresponding \emph{dual basis} $\beta_0,\beta_1, \ldots,\beta_{m-1}$ is defined by
\begin{align}
    \text{Tr}(\alpha^i\beta_j) \coloneqq \begin{cases}
1 & \text{if }i=j,\\
0 & \text{if }i\neq j,
\end{cases}
\end{align}
where the \emph{trace} $\text{Tr} \colon \mathbb{F}_{2^m} \rightarrow \mathbb{F}_2$ is the $\mathbb{F}_2$ linear map
\begin{align}
\text{Tr}(x) \coloneqq x + x^2 + \ldots + x^{2^{m-1}}.
\end{align}

Given a field element $a\in\mathbb{F}_{2^m}$, we will write
\begin{align}
    a=\sum_{i=0}^{m-1} \lceil a \rceil_i \alpha^i = \sum_{i=0}^{m-1} \lfloor a \rfloor_i \beta_i .
\end{align}
The binary row vector $\lceil a \rceil \in\{0,1\}^m$ represents the coefficients of $a$ in the primal basis, and the binary row vector $\lfloor a \rfloor \in\{0,1\}^m$ represents the coefficients of $a$ in the dual basis.

The trace is linear over $\mathbb{F}_2$ and the binary symmetric matrix $W$ given by
\begin{align}
\label{eq:W}
W_{ij} \coloneqq \text{Tr}\left( \alpha^i \alpha^j \right), \ i,j = 0,1,\ldots,m-1
\end{align}
satisfies
\begin{align}
    \lfloor a \rfloor = \lceil a \rceil W,
\end{align}
thereby translating primal coordinates to dual coordinates.

The trace inner product $\text{Tr}(ab)=\langle a,b \rangle_{\text{tr}}$ is given by 
\begin{align}
\label{eq:tr_ab}
    \text{Tr}(ab) = \lceil a \rceil W \lceil b \rceil^\mathsf{T} = \lceil a \rceil \cdot \lfloor b \rfloor \ (\text{mod }2).
\end{align}

The matrix $W$ is non-singular since the trace inner product is non-degenerate (if $\text{Tr}(xz) = 0$ for all $z \in \mathbb{F}_{2^m}$ then $x = 0$).
Observe that $W$ is a Hankel matrix, since if $i+j = h+k$ then $\text{Tr}(\alpha^i \alpha^j) = \text{Tr}(\alpha^h \alpha^k)$.


\subsection{The Pauli Group}


The single qubit (Hermitian) \emph{Pauli} matrices are
\begin{align}
I_2 \coloneqq 
\begin{bmatrix}
1 & 0 \\
0 & 1
\end{bmatrix} , \ 
X \coloneqq 
\begin{bmatrix}
0 & 1 \\
1 & 0
\end{bmatrix} , \ 
Z \coloneqq 
\begin{bmatrix}
1 & 0 \\
0 & -1
\end{bmatrix} , \  
Y \coloneqq i XZ = 
\begin{bmatrix}
0 & -i \\
i & 0
\end{bmatrix},
\end{align} 
where $i \coloneqq \sqrt{-1}$ and $I_2$ is the $2 \times 2$ identity matrix~\cite{Nielsen-2010}.


Pauli matrices on $m$ qubits are described by Kronecker products of $m$ single qubit Pauli matrices.
We associate with each pair of finite field elements $(a,b)\in\mathbb{F}_{2^m} \times \mathbb{F}_{2^m}$ the $m$-fold Kronecker product
\begin{align}
\label{eq:d_ab}
D(a, b) \coloneqq X^{\lceil a\rceil_0} Z^{\lfloor b\rfloor_0} \otimes \cdots \otimes X^{\lceil a\rceil_{m-1}} Z^{\lfloor b\rfloor_{m-1}} \in \mathbb{U}_N,
\end{align}
where $N \coloneqq 2^m$ and $\mathbb{U}_N$ denotes the group of all $N \times N$ unitary operators.

The \emph{$m$-qubit Pauli group} $P_N$ (also called the \emph{Heisenberg-Weyl group}) consists of all operators $i^{\kappa} D(a,b)$. The values $i^{\kappa}$, where $\kappa \in \mathbb{Z}_4 \coloneqq \{0,1,2,3\}$, are called \emph{quaternary phases}. 
The order $|P_N| = 4N^2$ and the \emph{center} of this group is $\langle i I_N \rangle \coloneqq \{ I_N, i I_N, -I_N, -i I_N \}$, where $I_N$ is the $N \times N$ identity matrix.
Hence, the homomorphism $\psi \colon P_N \rightarrow \mathbb{F}_2^{2m}$ defined by
\begin{align}
\label{eq:psi}
\psi(i^{\kappa} D(a, b)) \coloneqq [\lceil a\rceil, \lfloor b\rfloor] \ \forall \ \kappa \in \mathbb{Z}_4
\end{align}
has kernel $\langle i I_N \rangle$ and allows us to represent elements of $P_N$ (up to multiplication by scalars) as binary row vectors or pairs of $\mathbb{F}_{2^m}$ elements.

Multiplication in $P_N$ satisfies the identity
\begin{align}
\label{eq:P_N_commute}
&D(a, b) D(c, d) = (-1)^{\lceil c\rceil \cdot \lfloor b\rfloor + \lceil a\rceil \cdot \lfloor d\rfloor } D(c,d) D(a, b) .
\end{align}
The standard \emph{symplectic inner product} in $\mathbb{F}_2^{2m}$ is defined as 
\begin{align}
\syminn{[\lceil a\rceil,\lfloor b\rfloor]}{[\lceil c\rceil,\lfloor d\rfloor]} \coloneqq & \lceil c\rceil \cdot \lfloor b\rfloor + \lceil a\rceil \cdot \lfloor d\rfloor \nonumber\\
=& [\lceil a\rceil,\lfloor b\rfloor]\ \Omega \ [\lceil c\rceil,\lfloor d\rfloor]^\mathsf{T} ,
\end{align}
where the symplectic form $\Omega \coloneqq 
\begin{bmatrix}
0 & I_m \\ 
I_m & 0
\end{bmatrix}$ (see~\cite{Calderbank-it98*2,Rengaswamy-arxiv18,Rengaswamy-tqe20}).
It follows from~\eqref{eq:tr_ab} that
\begin{align}
    \syminn{[\lceil a\rceil,\lfloor b\rfloor]}{[\lceil c\rceil,\lfloor d\rfloor]} = \text{Tr}(ad+bc).
\end{align}
Therefore, two operators $D(a, b)$ and $D(c, d)$ commute if and only if $\text{Tr}(ad+bc) = 0$.

\subsection{The Clifford Group}

The \emph{Clifford group} $\text{Cliff}_N$ is the \emph{normalizer} of $P_N$ in the unitary group $\mathbb{U}_N$. It consists of all unitary matrices $g \in \mathbb{C}^{N \times N}$ for which $g D(a, b) g^{\dagger} \in P_N$ for all $D(a,b) \in P_N$, where $g^{\dagger}$ is the Hermitian transpose of $g$~\cite{Gottesman-arxiv09}. 

The Clifford group contains $P_N$ and has size $|\text{Cliff}_N| =  2^{m^2 + 2m} \prod_{j=1}^{m} (4^j - 1)$ (ignoring scalars $e^{2\pi i \theta}, \theta \in \mathbb{R}$)~\cite{Calderbank-it98*2}.
Every operator $g \in \text{Cliff}_N$ induces an automorphism of $P_N$ by conjugation.
Note that the inner automorphisms induced by matrices in $P_N$ preserve every conjugacy class $\{ \pm D(a,b) \}$ and $\{ \pm i D(a,b) \}$, because~\eqref{eq:P_N_commute} implies that elements in $P_N$ either commute or anti-commute.
Matrices $D(a,b)$ are symmetric or anti-symmetric according as $\text{Tr}(ab) = 0$ or $1$, hence the matrix
\begin{align}
E(a,b) \coloneqq i^{\text{Tr}(ab)} D(a,b)
\end{align}
is Hermitian.
Note that $E(a,b)^2 = I_N$.

The automorphism induced by a Clifford element $g$ satisfies
\begin{align}
\label{eq:symp_action}
g E(a,b) g^{\dagger} = \pm E\left( [\lceil a\rceil,\lfloor b\rfloor] F_g \right), {\rm where} \ F_g = 
\begin{bmatrix}
A_g & B_g \\
C_g & D_g
\end{bmatrix} 
\end{align}
is a $2m \times 2m$ binary matrix that preserves symplectic inner products:
\begin{align}
    \syminn{[\lceil a\rceil,\lfloor b\rfloor]F_g}{[\lceil c\rceil,\lfloor d\rfloor]F_g} = \syminn{[\lceil a\rceil,\lfloor b\rfloor]}{[\lceil c\rceil,\lfloor d\rfloor]} .
\end{align}
Hence, $F_g$ is called a \emph{binary symplectic matrix} and the symplectic property reduces to $F_g \Omega F_g^\mathsf{T} = \Omega$, or equivalently 
\begin{align}
\label{eq:symp_conditions}
A_g B_g^\mathsf{T} = B_g A_g^\mathsf{T}, \ C_g D_g^\mathsf{T} = D_g C_g^\mathsf{T}, \ A_g D_g^\mathsf{T} + B_g C_g^\mathsf{T} = I_m .
\end{align}
The symplectic property encodes the fact that the automorphism induced by $g$ must respect commutativity in $P_N$.
Let $\text{Sp}(2m,\mathbb{F}_2)$ denote the group of symplectic $2m \times 2m$ matrices over $\mathbb{F}_2$.
The map $\phi \colon \text{Cliff}_N \rightarrow \text{Sp}(2m,\mathbb{F}_2)$ defined by
\begin{align}
\phi(g) \coloneqq F_g 
\end{align}
is a homomorphism with kernel $P_N$, and every Clifford operator projects onto a symplectic matrix $F_g$.
Thus, $P_N$ is a normal subgroup of $\text{Cliff}_N$ and $\text{Cliff}_N/P_N \cong \text{Sp}(2m,\mathbb{F}_2)$.
This implies that $|\text{Sp}(2m,\mathbb{F}_2)| = 2^{m^2} \prod_{j=1}^{m} (4^j - 1)$ (also see~\cite{Calderbank-it98*2}).

Table~\ref{tab:std_symp} lists elementary symplectic transformations $F_g$, that generate the binary symplectic group $\text{Sp}(2m,\mathbb{F}_2)$, and the corresponding unitary automorphisms $g \in \text{Cliff}_N$, which together with $P_N$ generate $\text{Cliff}_N$.
(See~\cite[Appendix I]{Rengaswamy-arxiv18} for a discussion on the Clifford gates and circuits corresponding to these transformations.)

\begin{table}
\caption[caption]{\label{tab:std_symp}A generating set of symplectic matrices and their corresponding unitary operators.\\ {\normalfont The number of $1$s in $Q$ and $P$ directly relates to number of gates involved in the circuit realizing the respective unitary operators (see~\cite[Appendix I]{Rengaswamy-arxiv18,Rengaswamy-tqe20}). The $N$ coordinates are indexed by binary vectors $v \in \mathbb{F}_2^m$, and $e_v$ denotes the standard basis vector in $\mathbb{C}^N$ with an entry $1$ in position $v$ and all other entries $0$. Here $H_{2^t}$ denotes the Walsh-Hadamard matrix of size $2^t$, $U_t = {\rm diag}\left( I_t, 0_{m-t} \right)$ and $L_{m-t} = {\rm diag}\left( 0_t, I_{m-t} \right)$. }}
\centering
\vspace{-0.1cm}
\begin{tabular}{c|c}
Symplectic Matrix $F_g$ & Clifford Operator $g$\\
\hline
~ & ~ \\
$\Omega = \begin{bmatrix} 0 & I_m \\ I_m & 0 \end{bmatrix}$ & $H_N = H_2^{\otimes m}$ \\
~ & ~ \\
$L_Q = \begin{bmatrix} Q & 0 \\ 0 & Q^{-\mathsf{T}} \end{bmatrix}$ & $\ell_Q: e_{v} \mapsto e_{v Q}$ \\
~ & ~ \\
$T_P = \begin{bmatrix} I_m & P \\ 0 & I_m \end{bmatrix}; P = P^\mathsf{T}$ & $t_P\ =\ {\rm diag}\left( i^{v P v^\mathsf{T} \bmod 4} \right)$ \\
~ & ~ \\
$G_t = \begin{bmatrix} L_{m-t} & U_t \\ U_t & L_{m-t} \end{bmatrix}$ & $g_t = H_{2^t} \otimes I_{2^{m-t}}$ \\
~ & ~ \\
\hline
\end{tabular}
\vspace{-0.4cm}
\end{table}

\subsection{Symplectic Transvections}

The symplectic transvection $\tau_h$ corresponding to the row vector $h\in\mathbb{F}_2^{2m}$ is the map
$\tau_h \colon \mathbb{F}_2^{2m}\to \mathbb{F}_2^{2m}$ defined by
\begin{align}
    \tau_h(x)\coloneqq x + \syminn{x}{h}h = xZ_h,
\end{align}
where $Z_h\coloneqq I_{2m}+\Omega h^\mathsf{T} h$ is the associated symplectic matrix of $\tau_h$.

The set of symplectic transvections is a conjugacy class in Sp$(2m, \mathbb{F}_2)$, since if $h=[h_1, h_2]$ with $h_1,h_2\in\mathbb{F}_2^m$, then
\begin{align}
    \label{eq:transvec_conjugacy}
    F_g^{-1}Z_hF_g=Z_{[h_1A+h_2C,h_1B+h_2D]}=Z_{hF_g},
\end{align}
where
\begin{align*}
    F_g=\begin{bmatrix}
    A_g & B_g \\
    C_g & D_g
    \end{bmatrix} \text{, and }
    F_g^{-1}=\begin{bmatrix}
    D_g^\intercal & B_g^\intercal \\
    C_g^\intercal & A_g^\intercal
    \end{bmatrix}.
\end{align*}

Hence, the group generated by the symplectic transvections is normal in Sp$(2m,\mathbb{F}_2)$. Since Sp$(2m,\mathbb{F}_2)$ is simple, it is generated by the set of symplectic transvections. 

We shall view the row vectors that define a symplectic transvection as elements of the finite field $\mathbb{F}_{2^m}$. Given $a,b,h_1,h_2\in\mathbb{F}_{2^m}$, let $h = [\lceil h_1\rceil, \lfloor h_2 \rfloor]$, then
\begin{align}
&Z_h([\lceil a\rceil, \lfloor b \rfloor]) \nonumber\\
=& [\lceil a\rceil, \lfloor b \rfloor] + (\lceil a\rceil \cdot \lfloor h_2 \rfloor + \lceil h_1\rceil \cdot \lfloor b \rfloor)[\lceil h_1\rceil, \lfloor h_2 \rfloor] \nonumber\\
\label{eq:tran}
\equiv &
    \begin{pmatrix}
    a & b
    \end{pmatrix} + \text{Tr}(ah_2+bh_1)\begin{pmatrix}
    h_1 & h_2
    \end{pmatrix}
\end{align}
We will often write the transvection $Z_h$ as $Z_{(h_1,h_2)}$, where $h_1,h_2\in\mathbb{F}_{2^m}$.

\section{The Pauli Geometry}\label{sec:pauli_geometry}

In this section, we will define a graph on Pauli matrices, where Clifford elements act as graph automorphisms. 
We will build on prior work~\cite{Webb-arxiv16} showing that a set of automorphisms forms a unitary $2$-design if it acts transitively on vertices, or a unitary $3$-design if it acts transitively on vertices, on edges, and on non-edges.

\subsection{The Maximal Commutative Subgroups}
To begin with, we shall review some concepts that are closely related to the Pauli group: stabilizers, maximal commutative subgroups, and stabilizer states.

A \emph{stabilizer} group is a subgroup $S$ of $P_N$ generated by commuting Hermitian matrices of the form $\pm E(a,b)$, with the additional property that if $E(a,b) \in S$ then $-E(a,b) \notin S$~\cite[Chapter 10]{Nielsen-2010}.
The operators $\frac{I_N \pm E(a,b)}{2}$ project onto the $\pm 1$ eigenspaces of $E(a,b)$, respectively.

Since all elements of $S$ are unitary, Hermitian and commute with each other, they can be diagonalized simultaneously with respect to a common orthonormal basis, and their eigenvalues are $\pm 1$ with algebraic multiplicity $N/2$.
We refer to such a basis as the \emph{common eigenbasis} or simply the \emph{eigenbasis} of the subgroup $S$, and to the subspace of eigenvectors with eigenvalue $+1$ as the $\emph{+1}$ \emph{eigenspace} of $S$.

If the subgroup $S$ is generated by $E(a_i,b_i), i = 1,\ldots,k$, then the operator
\begin{align}
\frac{1}{2^k} \prod_{i=1}^{k} (I_N + E(a_i,b_i))
\end{align}
projects onto the $2^{m-k}$-dimensional subspace $V(S)$ fixed pointwise by $S$, i.e., the $+1$ eigenspace of $S$.
The subspace $V(S)$ is the \emph{stabilizer code} determined by $S$.
One uses the notation $\llbr m,m-k \rrbr$ code to represent that $V(S)$ encodes $m-k$ \emph{logical} qubits into $m$ \emph{physical} qubits.

Let $\psi(S)$ denote the subspace of $\mathbb{F}_2^{2m}$ formed by the binary representations of the elements of $S$ using the homomorphism $\psi$ in~\eqref{eq:psi}.
A generator matrix for $\psi(S)$ is
\begin{align}
\label{eq:stabilizer_G}
G_S \coloneqq [\lceil a_i \rceil, \lfloor b_i\rfloor]_{i = 1,\ldots,k} \in \mathbb{F}_2^{k \times 2m}\ \text{s.t.}\ G_S\ \Omega\ G_S^\mathsf{T} = 0,
\end{align}
where $0$ is the $k \times k$ matrix with all entries zero.

A stabilizer group $S$ defined by $k = m$ generators is called a \emph{maximal commutative subgroup} of $P_N$ and $\psi(S)$ is called a \emph{maximal isotropic subspace} of $\mathbb{F}_2^{2m}$.
The generator matrix $G_S$ has rank $m$ and can be row-reduced to $[0 \mid I_m]$ if $S = Z_N \coloneqq \{ E(0,b) \colon b \in \mathbb{F}_{2^m} \}$, or to the form $[I_m \mid P]$ if $S$ is disjoint from $Z_N$. The condition $G_S \Omega G_S^\mathsf{T} = 0$ implies $P = P^\mathsf{T}$.
\begin{remark}
We will denote these maximal commutative subgroups as $E([0 \mid I_m])$ and $E([I_m \mid P])$, respectively, and $E([I_m \mid 0]) = X_N \coloneqq \{ E(a,0) \colon a \in \mathbb{F}_{2^m} \}$. Notice that we still employ the form of $E(a,b)$, where $a,b\in\mathbb{F}_{2^m}$, to represent elements inside of each subgroup.
\end{remark}


\subsection{The Kerdock Set}

Kerdock sets connect Kerdock codes~\cite{Hammons-it94} with maximal commutative subgroups of the Pauli group.
We refer the reader to~\cite{Can-arxiv19} for more information about connections between binary and quaternary Kerdock codes and eigenbases of maximal commutative subgroups $E( [ I_m \mid P_z ] )$, where $z \in \mathbb{F}_{2^m}$.

We write multiplication by $z \in \mathbb{F}_{2^m}$ as a linear transformation $xz \equiv \lceil x \rceil A_z$.
For $z = 0, A_0 = 0$, and for $z = \alpha^i$ the matrix $A_z = A^i$ for $i = 0,1,\ldots,2^m - 2$, where $A$ is the matrix that represents multiplication by the primitive element $\alpha$.
The matrix $A$ is the \emph{companion matrix} of the primitive irreducible polynomial $p(x) = p_0 + p_1 x + \ldots + p_{m-1} x^{m-1} + x^m$ over the binary field.
Thus,
\begin{align}
\label{eq:A}
A \coloneqq 
\begin{bmatrix}
0 & 1 & 0 & \cdots & 0 \\
0 & 0 & 1 & \cdots & 0 \\
  & \vdots  &   & \ddots & \vdots \\
0 & 0 & 0 & \cdots & 1 \\
p_0 & p_1 & p_2 & \cdots & p_{m-1} 
\end{bmatrix}.
\end{align}

\begin{lemma} \label{lem:f2mrelations}
It follows directly from the arithmetic of $\mathbb{F}_{2^m}$ that the matrices $A_z$ and $W$ satisfy:
\begin{enumerate}
    \item[(a)] $A_z A_x = A_x A_z = A_{xz}$;
    \item[(b)] $A_x + A_z = A_{x+z}$;
    \item[(c)] $A_z W = W A_z^\mathsf{T}$;
\end{enumerate}
\end{lemma}
\begin{proof} 
Specifically, for (c), observe that 
\begin{align*}
(\lceil x\rceil A_z) W \lceil y \rceil^\mathsf{T} = \text{Tr}((xz)y) = \text{Tr}(x(yz)) = \lceil x \rceil W (\lceil y \rceil A_z)^\mathsf{T}.  
\end{align*}
The other two properties are easily verified. \qed
\end{proof}

The \emph{Kerdock set} $P_K(m)$ consists of all symmetric matrices $P_z=A_zW$, where $z\in\mathbb{F}_{2^m}$~\cite{Can-arxiv19}. 
It is closed under binary addition, and if $P,Q \in P_{\text{K}}(m)$ are distinct, then $P+Q$ is non-singular, which in turn implies that the maximal commutative subgroups $E([I_m \mid P_z])$ determined by the Kerdock matrices $P_z$ intersect trivially.
Together with $Z_N = E([0 \mid I_m])$, they partition all $(N^2 - 1)$ non-identity Hermitian Pauli matrices.
Hence, given a non-identity Hermitian Pauli matrix $E(a,b)$, it follows that there is a sign $\epsilon \in \{ \pm 1 \}$ such that $\epsilon E(a,b)$ is in one of the $N+1$ subgroups determined by all $P_z \in P_{\text{K}}(m)$ and $Z_N$.
If $E(a, b)$ is in the maximal commutative subgroup $E([I_m \mid P_z])$, then we must have
\begin{align*}
    \lfloor b \rfloor = \lceil a \rceil A_z W = \lceil az \rceil W =  \lfloor az \rfloor \Rightarrow b=az.
\end{align*}
Therefore, each element in $E([I_m \mid P_z])$ or $E([0 \mid I_m])$ can be written as $E(a,az)$ or $E(0,a)$, respectively, for some $a\in\mathbb{F}_{2^m}$.

\begin{remark}
Throughout this paper, we will illustrate theory developed for general $m$ by reducing to the special case $m=3$. 
We constructed $\mathbb{F}_8$ by adjoining a root $\alpha$ of $p(x)=x^3+x+1$ over $\mathbb{F}_2$. Then, we have $\mathbb{F}_8 = \{0,1,\alpha,\alpha^2, \ldots, \alpha^6\}$ and $\alpha^7 = 1$. The trace of $0, \alpha, \alpha^2,\alpha^4$ is $0$ and the trace of $1, \alpha^3, \alpha^5,\alpha^6$ is $1$. 
\end{remark}

\begin{example}
It follows from~\eqref{eq:A} and~\eqref{eq:W} that
\begin{align}
\label{eq:A&W}
A_\alpha = A^1 =
\begin{bmatrix}
0 & 1 & 0 \\
0 & 0 & 1\\
1 & 1 & 0
\end{bmatrix} \text{ and } \edit{W =
\begin{bmatrix}
1 & 0 & 0 \\
0 & 0 & 1\\
0 & 1 & 0
\end{bmatrix}}.
\end{align}
We may form $P_z=A_z W$ for any $z\in\mathbb{F}_{2^m}$ to generate the entire Kerdock set $P_K(m)$.
\end{example}


\subsection{The Pauli Graph}
\begin{definition}
\edit{The \emph{Pauli graph} $\mathbb{P}_N$ has $N^2-1$ vertices, each of which is labeled by pairs $\pm E(a,b)$ with $(a,b) \neq (0,0)$ and represented as $(a,b)$. A directed edge connects vertex $(a,b)$ to vertex $(c,d)$ if $E(a,b)$ commutes with $E(c,d)$, represented as a $2\times 2$ matrix $\begin{psmallmatrix}
a & b \\
c & d
\end{psmallmatrix}$
over $\mathbb{F}_{2^m}$.}
\end{definition}


\edit{Apparently, if there is a directed edge connecting $(a, b)$ to $(c, d)$, then $(c, d)$ must also connect to $(a, b)$. The need for this deterministic additional edge between commuting Paulis will be clear when we discuss orbits in Section \ref{sec:orbits} and our Markov process in Section \ref{sec:markov}.} Then, it follows from~\eqref{eq:tr_ab} that
\begin{align}
\label{eq:tr_edge}
\text{Tr}\left(\text{det}\begin{pmatrix}
a & b \\
c & d
\end{pmatrix}\right) = \text{Tr}( ad + bc) = 0.
\end{align}

We shall distinguish two types of \edit{directed edges} in $\mathbb{P}_N$:
\begin{enumerate}
    
\item \emph{Type-$1$} edges connect vertices from the same maximal commutative subgroup $E([I_m \mid P_z]), z \in \mathbb{F}_{2^m}$, or from $E([0 \mid I_m])$. The determinant of a type-1 edge matrix is $0$:
\begin{align}
\label{eq:type1_edg_mat}
    \text{det}\begin{pmatrix}
a & az \\
b & bz
\end{pmatrix} = 0, \text{ where }a\neq b, ab\neq 0.
\end{align}

\item \emph{Type-$2$} edges connect vertices from different maximal commutative subgroups. The determinant of a type-2 edge matrix is a non-zero field element with trace $0$:
\begin{align}
\label{eq:type2_edg_mat}
    \text{det}\begin{pmatrix}
a & az_1 \\
b & bz_2
\end{pmatrix} = ab(z_1+z_2) \neq 0, \text{ where }z_1\neq z_2, ab\neq 0.
\end{align}

\end{enumerate}

The \edit{determinant of a} $2 \times 2$ matrix representing \edit{an ordered} pair of non-commuting Pauli matrices \edit{is a non-zero field element that} has trace equal to 1.
We shall refer to an arbitrary $2\times 2$ matrix \edit{over $\mathbb{F}_{2^m}$, which represents an edge or a non-edge,} as a Pauli pair matrix. 

\begin{lemma}
The Pauli graph $\mathbb{P}_N$ is a strongly regular graph with parameters
\begin{align}
n = N^2 - 1, \ t = \frac{N^2}{2} - 2, \ \lambda = \frac{N^2}{4} - 3, \ \mu = \frac{N^2}{4} - 1,
\end{align}
where $n$ is the number of vertices, $t$ is the \edit{in-degree and out-degree} of each vertex, and $\lambda$ or $\mu$ is the number of vertices joined to a pair of distinct vertices $x, y$ according as $x, y$ are joined or not joined respectively \cite[Definiton 2.4]{Cameron-1991}.
\end{lemma}
\begin{proof}
A vertex $(c,d)$ joined to a given vertex $(a,b)$ must satisfy $\text{Tr}(ad+bc)=0$ and half of the elements in $\mathbb{F}_{2^m}$ have trace $0$. For each possible value of $ad+bc$, there exists $N$ pairs $(c,d)$ that are feasible. 
After eliminating the solutions $(0,0)$ and $(a,b)$, we are left with $t = \frac{N^2}{2} - 2$ distinct vertices $(c,d)$ joined to $(a,b)$.

Given vertices $(a,b)$ and  $(c,d)$, a vertex $(e,f)$ joined to both $(a,b)$ and $(c,d)$ satisfies $\text{Tr}(af+be)=0$ and $\text{Tr}(cf+de)=0$. Each pair of $af+be$ and $cf+de$ values have one solution for $(e,f)$.
When $(a,b)$ is not joined to $(c,d)$, we only need to eliminate the solution $(0,0)$.
When $(a,b)$ is joined to $(c,d)$ we need to eliminate $(0,0), (a,b)$ and $(c,d)$. \qed
\end{proof}

\edit{The number of edges in $\mathbb{P}_N$ is $(N^2 - 1)(N^2 - 4)/2$.
The number of type-$1$ edges is $(N^2-1)(N-2)$ and the number of type-$2$ edges is $N(N^2-1)(N-2)/2$. }


Using these properties of the Pauli graph, we can now discuss unitary designs.
We denote the linear operators acting on a complex Euclidean space $\mathcal{X}$ (e.g., $\mathcal{X} = \mathbb{C}^N$) as $L(\mathcal{X})$ and the quantum channels acting on $L(\mathcal{X})$ as $C(\mathcal{X})$.
\begin{definition}
Let $k$ be a positive integer and $\mathcal{E} = \{\alpha_i, U_i\}_{i=1}^{n}$ be an ensemble where the unitary matrix $U_i$ is selected with probability $\alpha_i$. The superoperators $\mathcal{G}_\mathcal{E}, \mathcal{G}_H \in C((\mathbb{C}^N)^{\otimes k})$ are given by
\begin{align}
    \mathcal{G}_\mathcal{E} (\rho) & = \sum_i \alpha_i U_i^{\otimes k}\rho (U_i^\dagger)^{\otimes k}, \\
    \mathcal{G}_H (\rho) & = \int_{\mathbb{U}_N} d\eta(U) U^{\otimes k}\rho (U^\dagger)^{\otimes k},
\end{align}
where $\eta(\cdot)$ represents the Haar measure on the unitary group $\mathbb{U}_N$. The ensemble $\mathcal{E}$ is a \emph{unitary $k$-design} if and only if $\mathcal{G}_\mathcal{E} = \mathcal{G}_H$. The linear transformations determined by $\mathcal{G}_\mathcal{E}$ and $\mathcal{G}_H$ are called \emph{$k$-fold twirls}. A unitary $k$-design is defined by the property that the ensemble twirl coincides with the full unitary twirl. 
\end{definition}

Notice that elements of the Clifford group act by conjugation on $P_N$, inducing automorphisms of the graph $\mathbb{P}_N$. It is well known that the symplectic group $\text{Sp}(2m, \mathbb{F}_2)$ acts transitively on vertices, on edges, and on non-edges of $\mathbb{P}_N$.

Following Webb~\cite{Webb-arxiv16}, we say that an ensemble $\mathcal{E} = \{\alpha_i, U_i\}_{i=1}^{n}$ of Clifford elements $U_i$ is \emph{Pauli mixing} if for every vertex $(a,b)$ the distribution $\{ \alpha_i, U_i E(a,b) U_i^{\dagger} \}$ is uniform over vertices of $\mathbb{P}_N$.
The ensemble $\mathcal{E}$ is \emph{Pauli $2$-mixing} if it is Pauli mixing and if for every edge (resp. non-edge)
$\begin{psmallmatrix}
a & b \\
c & d
\end{psmallmatrix}$, the distribution $$\{ \alpha_i, (U_i E(a,b) U_i^{\dagger}, U_i E(c,d) U_i^{\dagger}) \}$$ is uniform over edges (resp. non-edges) of $\mathbb{P}_N$.

\begin{theorem}
\label{thm:uni-design}
Let $G$ be a subgroup of the Clifford group containing all $D(a,b) \in HW_N$, and let $\mathcal{E} = \{ \frac{1}{|G|}, U \}_{U \in G}$ be the ensemble defined by the uniform distribution.
If $G$ acts transitively on vertices of $\mathbb{P}_N$, then $\mathcal{E}$ is a unitary $2$-design, and if $G$ acts transitively on vertices, edges and non-edges, then $\mathcal{E}$ is a unitary $3$-design.
\end{theorem}
\begin{proof}
Transitivity means a single orbit so that random sampling from $G$ results in the uniform distribution on vertices, edges, and non-edges.
Hence, transitivity on vertices implies $\mathcal{E}$ is Pauli mixing and transitivity on vertices, edges and non-edges implies $\mathcal{E}$ is Pauli $2$-mixing.
It now follows from~\cite{Webb-arxiv16} or~\cite{Cleve-arxiv16} that Pauli mixing (resp. Pauli $2$-mixing) implies $\mathcal{E}$ is a unitary $2$-design (resp. unitary $3$-design). \qed
\end{proof}

Theorem~\ref{thm:uni-design} implies that random sampling from the Clifford group gives an exact unitary $3$-design.

\section{The Unitary 2-Design}\label{sec:unitary_2}

The unitary 2-design we consider in this paper is a symplectic subgroup isomorphic to $\text{PSL}(2, 2^m)$, and different descriptions may be found in~\cite{Chau-ieee05},~\cite{Cleve-arxiv16}, and~\cite{Can-arxiv19}. 
We will follow the perspective provided in~\cite{Can-arxiv19}.
We first show that it acts transitively on vertices of the graph $\mathbb{P}_N$.
Then, we use our finite field representation to describe how this subgroup acts on \edit{ordered} pairs of commuting and anti-commuting Pauli Matrices.

\subsection{The Symplectic Subgroup Isomorphic to PSL$(2, 2^m)$}
\label{sec:psl}
We first introduce PSL$(2, 2^m)$ and realize each transformation as a symplectic matrix. 
Then, we explain why this symplectic subgroup forms a unitary $2$-design by showing how the group elements permute the maximal commutative subgroups $E([ I_m| P_z])$ and $E([ 0 | I_m])$ and elements within each subgroup. 

The \emph{projective special linear group} of $2\times 2$ matrices over $\mathbb{F}_{2^m}$ is defined as 
\begin{align}
    \text{PSL}(2, 2^m) \coloneqq \left\{
    \begin{pmatrix}
    \alpha & \beta \\
    \gamma & \delta
    \end{pmatrix}: \alpha, \beta, \gamma, \delta \in \mathbb{F}_{2^m}; \ \alpha\delta+\beta\gamma=1
    \right\}.
\end{align}
The order $|\text{PSL}(2, 2^m)|=(N+1)N(N-1)=2^{3m}-2^m$. The action of each $2\times 2$ matrix 
$\begin{psmallmatrix}
\alpha & \beta \\
\gamma & \delta
\end{psmallmatrix}$ over $\mathbb{F}_{2^m}$ on $1$-dimensional subspaces of $\mathbb{F}_{2^m} \times \mathbb{F}_{2^m}$ is associated with a transformation
\begin{align*}
    f(z) = \frac{\beta +\delta z}{\alpha +\gamma z} 
\end{align*}
acting on the projective line $\mathbb{F}_{2^m} \cup \{ \infty \}$, given that
\begin{align}
\begin{pmatrix}
1 & z
\end{pmatrix}
\begin{pmatrix}
\alpha & \beta \\
\gamma & \delta
\end{pmatrix}= 
\begin{pmatrix}
\alpha +\gamma z & \beta +\delta z
\end{pmatrix} \equiv
\begin{pmatrix}
0 & 1
\end{pmatrix}
\text{ or }
\begin{pmatrix}
1 & \frac{\beta +\delta z}{\alpha +\gamma z}
\end{pmatrix}.
\end{align}

The group $\text{PSL}(2,2^m)$ is generated by the transformations $z \mapsto z+x, z \mapsto zx$, and $z \mapsto 1/z$.
We realize each of these transformations as a symplectic transformation.
We recall that $A_z W A_z^\mathsf{T} = A_z^2 W$ from part (c) of Lemma~\ref{lem:f2mrelations}, and for convenience we work with maximal commutative subgroups $E([I_m \mid A_z^2 W])$, i.e., the Kerdock matrices are $P_z = A_z^2 W$.
Note that every field element $\beta \in \mathbb{F}_{2^m}$ is a square, so this is equivalent to $P_z = A_z W$.


\begin{enumerate}
\item[(a)] $z \mapsto z+x \ \text{becomes}\ [I_m \mid A_z^2 W] \mapsto [I_m \mid A_{x+z}^2 W] \colon$ 
\begin{align}
[I_m \mid A_z^2 W] 
\begin{bmatrix} 
I_m & A_x^2 W \\ 
0 & I_m 
\end{bmatrix} & = [I_m \mid (A_z^2 + A_x^2) W] \nonumber \\
\label{eq:Pk1}
              & \equiv [I_m \mid (A_{x+z}^2) W].
\end{align}
\end{enumerate}

\begin{example}
When $m=3$, $x=\alpha$, and $z=\alpha^3$, we have $x+z=1$, $A_x = A$, and $A_z = A^3$. Then, 
\begin{align}
A^2_x + A^2_z = A^2_{x+z} = I_3 = A^0.
\end{align}
\end{example}

\begin{enumerate}
\item[(b)] $z \mapsto xz \ \text{becomes}\ [I_m \mid A_z^2 W] \mapsto [I_m \mid A_{xz}^2 W] \colon$
\begin{align}
[I_m \mid A_z^2 W] 
\begin{bmatrix} 
A_x^{-1} & 0 \\ 
0 & A_x^\mathsf{T}
\end{bmatrix} & = [A_x^{-1} \mid A_z^2 W A_x^\mathsf{T}] \nonumber \\
              & = [A_x^{-1} \mid A_x A_z^2 W] \nonumber \\
\label{eq:Pk2}
              & \equiv [I_m \mid A_{xz}^2 W].
\end{align} 

\item[(c)] $z \mapsto 1/z \ \text{becomes}\ [I_m \mid A_z^2 W] \mapsto [I_m \mid A_{z^{-1}}^2 W] \colon$
\begin{align}
[I_m \mid A_z^2 W] 
\begin{bmatrix} 
0 & I_m \\ 
I_m & 0 
\end{bmatrix} 
\begin{bmatrix} 
W^{-1} & 0 \\ 
0 & W^\mathsf{T}
\end{bmatrix} & = [A_z^2 W \mid I_m] 
\begin{bmatrix} 
W^{-1} & 0 \\ 
0 & W 
\end{bmatrix} \nonumber \\
  & = [A_z^2 \mid W] \nonumber \\
\label{eq:Pk3}
  & \equiv [I_m \mid A_{z^{-1}}^2 W].
\end{align} 
Note that if we start with $z = 0$, i.e., the subgroup $E([I_m \mid 0])$, then since $W$ is invertible the final subgroup is $E([0 \mid I_m])$, interpreted as $z = \infty$.

\end{enumerate}

Therefore, PSL$(2,2^m)$ is isomorphic to $\mathfrak{P}_{m}$, a group of symplectic matrices defined as
\begin{align}
\label{eq:P_generators}
\mathfrak{P}_{m} & \coloneqq \langle T_{A_x^2 W},\ L_{A_x^{-1}},\ \Omega L_{W^{-1}} ; x \in \mathbb{F}_{2^m} \rangle \cong \text{PSL}(2,2^m),
\end{align}
and each PSL$(2,2^m)$ element induces a product of basic symplectic matrices in Table~\ref{tab:std_symp}. 
The isomorphism $\theta \colon \text{PSL}(2,2^m) \rightarrow \mathfrak{P}_{m}$ can be defined as
\begin{align}
\label{eq:PSL_iso}
\theta \left(
\begin{pmatrix}
\alpha & \beta \\
\gamma & \delta
\end{pmatrix} \right)  &\coloneqq T_{A_{\delta/ \gamma}^2 W} \cdot L_{A_\gamma^{-2}} \cdot \Omega L_{W^{-1}} \cdot T_{A_{\alpha/ \gamma}^2 W} \\
  & = \begin{bmatrix}
A_\delta^2 & A_\beta^2 W \\
W^{-1} A_\gamma^2 & (A_\alpha^2)^\mathsf{T}
\end{bmatrix},
\end{align}
where $\alpha,\beta,\gamma,\delta \in \mathbb{F}_{2^m}$ and $\alpha\delta + \beta\gamma = 1$ \edit{\cite[Lemma 23 and Corollary 24]{Can-arxiv19}}. 
The induced action on maximal commutative subgroups is given by
\begin{align}
    E([I_m\mid A_z^2W]) \mapsto E([I_m\mid A_{\frac{\beta +\delta z}{\alpha +\gamma z}}^2W]) \text{ or } E([ 0\mid I_m]).
\end{align}

Notice that the corresponding Clifford subgroup is larger than PSL$(2, 2^m)$ since $P_N$ forms the kernel of the homomorphism from $\text{Cliff}_N$ to Sp$(2m, \mathbb{F}_2)$. 

The first two factors in~\eqref{eq:PSL_iso} provide transitivity on the Hermitian matrices of all maximal commutative subgroups except $Z_N = E([0 \mid I_m])$, and \edit{the last two factors} enables exchanging any subgroup $E([I_m \mid P_z])$ with $E([0 \mid I_m])$.

To prove that PSL$(2,2^m)$ acts transitively on vertices of the Pauli graph $\mathbb{P}_N$, we only need to show that the group is transitive on a particular subgroup, say $E([I_m | 0])$. For any $(a,0)$ and $(b,0)$ where $a\neq b\in \mathbb{F}_{2^m}$, there always exists a group element 
$\begin{psmallmatrix}
a^{-1}b & 0 \\
0 & b^{-1}a
\end{psmallmatrix}$
that maps $(a,0)$ to $(b,0)$.

It then follows from Theorem~\ref{thm:uni-design} and~\eqref{eq:symp_action} that random sampling from the symplectic subgroup $\mathfrak{P}_{m}$ isomorphic to PSL$(2, 2^m)$ followed by a random Pauli matrix $D(a, b)$ produces a unitary $2$-design.


However, PSL$(2, 2^m)$ is only able to permute maximal commutative subgroups or elements within each subgroup. It is not transitive on edges of $\mathbb{P}_N$ since it fails to mix type-$1$ edges and type-$2$ edges. Thus, $\mathfrak{P}_{m}$ cannot be a unitary $3$-design.

\subsection{Orbit Invariants}\label{sec:orbits}

PSL$(2, 2^m)$ partitions the edges of the Pauli graph into orbits, and we now identify orbit invariants.

\begin{definition}\label{def:orbit_invar}
We calculate the \emph{orbit invariant} from any representative Pauli pair matrix $\begin{psmallmatrix}
a & b \\
c & d
\end{psmallmatrix}$ as follows:
\begin{enumerate}
    \item[(a)] The determinant $ad+bc$ is the orbit invariant for any non-edge matrix or type-2 edge matrix.
    \item[(b)] For any type-$1$ edge matrix, \edit{the first row is a scalar multiple of the second row}, given that the two vertices $(a,b)$ and $(c,d)$ are in the same maximal commutative subgroup. \edit{Its orbit invariant is the scalar 
    $\frac{a}{c}=\frac{b}{d}\in\mathbb{F}_{2^m}\setminus \{0, 1\}$. }
    
\end{enumerate}
\end{definition}

\begin{example}
When $m=3$, the following three Pauli pair matrices
\begin{align}
E_1=\begin{pmatrix}
\alpha & 0\\
1 & \alpha^2
\end{pmatrix},
E_2=\begin{pmatrix}
\alpha & 0\\
1 & \alpha
\end{pmatrix}, \text{and }
E_3=\begin{pmatrix}
\alpha & 0\\
1 & 0
\end{pmatrix}
\end{align}
have orbit invariants $\alpha^3$, $\alpha^2$, and \edit{$\alpha$. According to Definition \ref{def:orbit_invar}, \eqref{eq:tr_edge}, \eqref{eq:type1_edg_mat}, and~\eqref{eq:type2_edg_mat}, we can recognize that} $E_1$ is a non-edge, $E_2$ is a type-$2$ edge, and $E_3$ is a type-$1$ edge. 
\end{example}

\begin{theorem}
\label{thm:orb_inv}
Consider PSL$(2, 2^m)$ acting on Pauli pair matrices by right multiplication. Two Pauli pair matrices are in the same orbit if and only if they have the same orbit invariant. 
\end{theorem}
\begin{proof} 
Every matrix in PSL$(2, 2m)$ has determinant $1$, so Pauli pair matrices in the same orbit share the same determinant. If the determinant is $0$, then one row is a scalar multiple of the other and the scalar relation between rows is preserved by any linear transformation.

Consider two matrices, either with the same non-zero determinant, or with determinant $0$ and the same scalar relation between the two rows. There always exists a linear transformation with determinant $1$ that maps one to the other. \qed
\end{proof}

It follows directly from Theorem~\ref{thm:orb_inv} that we can use orbit invariants to represent and differentiate the orbits. We give some statistics about these orbits below:
\begin{enumerate}
    \item[(a)] There are $\frac{N}{2}$ non-edge orbits and there are an equal number of finite field elements with trace $1$. \edit{Each orbit has $(N^2-1)N$ elements. }
    \item[(b)] There are $\frac{N-2}{2}$ type-$2$ edge orbits and there are an equal number of non-zero finite field elements with trace $0$. \edit{Each orbit also has $(N^2-1)N$ elements.} 
    \item[(c)] \edit{There are $N-2$ type-$1$ edge orbits and there are an equal number of field elements in $\mathbb{F}_{2^m}\setminus \{0,1\}$.} Each orbit has $N^2-1$ elements. 
\end{enumerate}

\begin{example}
For $m=3$, there are $4$ non-edge orbits whose invariants are $1$, $\alpha^3$, $\alpha^5$, and $\alpha^6$; $3$ type-$2$ edge orbits whose invariants are $\alpha$, $\alpha^2$, and $\alpha^4$; and \edit{$6$ type-$1$ edge orbits whose invariants are $\alpha, \alpha^2, \cdots, \alpha^6$.}
\end{example}

\section{The Transvection Markov Process}\label{sec:markov}
We define a Markov process by applying a sequence of transvections to \edit{mix orbits, and a final $\text{PSL}(2, 2^m)$ element to mix edges or non-edges within each orbit.}
We claim that it gives an approximate unitary $3$-design by showing convergence to the uniform distribution on edges and on non-edges, in addition to the transitivity on vertices (Section~\ref{sec:psl}). 

Let \edit{$K=(N^2 - 1)(N^2 - 4)/2$} be the number of edges and \edit{$K'=(N^2 - 1)N^2/2$} be the number of non-edges in $\mathbb{P}_N$. Consider the underlying Markov chain on \edit{directed} edges (resp. non-edges) with a $K\times K$ (resp. $K' \times K'$) transition matrix. 
\edit{Since transvections generate the full Clifford group and Clifford elements act transitively on edges and non-edges, }the uniform distribution on all edges (resp. non-edges) is stationary. We are interested in the rates at which the two Markov processes converge to their corresponding stationary distributions. 

Sampling PSL$(2, 2^m)$ elements results in uniform probabilities within orbits.
Therefore, it suffices to reduce the two underlying Markov chains to orbits and only consider how \edit{non-identity} (symplectic) transvections transfer probability mass within the reduced state space. 
The dimensions of the new transition matrix on edge (resp. non-edge) orbits are \edit{$\frac{3}{2}(N-2)\times \frac{3}{2}(N-2)$} (resp. $\frac{N}{2}\times \frac{N}{2}$).


\subsection{The Transvection Markov Chain on Non-Edges}

\begin{theorem}
\label{thm:transition-non-edge}
Consider the Markov process with state space consisting of all non-edge orbits. 
The matrix $Q_1$ of state transition probabilities is given by
\begin{align}
\label{eq:transition-non-edge}
    Q_1=\frac{1}{4(N^2-1)}
    \begin{bmatrix}
    (N^2-4)I_{N/2} + 6NJ_{N/2}
    \end{bmatrix}\in \mathbb{R}^{\frac{N}{2} \times \frac{N}{2}},
\end{align}
where $I_{N/2}, J_{N/2}\in\mathbb{R}^{\frac{N}{2} \times \frac{N}{2}}$, $I_{N/2}$ is the identity matrix, and $J_{N/2}$ is the all ones matrix. 
\end{theorem}

\begin{proof}
We apply a random transvection $Z_{(h_1,h_2)}$, where $h_1,h_2\in\mathbb{F}_{2^m}$ and $(h_1,h_2)\neq (0,0)$, to a non-edge matrix 
$\begin{psmallmatrix}
a & 0\\
0 & b
\end{psmallmatrix}$ with orbit invariant $ab$, where Tr$(ab) = 1$.
According to~\eqref{eq:tran}, we have the following four cases:
\begin{enumerate}
    \item[(a)] Applying a transvection with Tr$(ah_2)=0$ and Tr$(bh_1)=0$ fixes the non-edge:
    \begin{align}
    \begin{pmatrix}
        a & 0\\
        0 & b
    \end{pmatrix}\xrightarrow{Z_{(h_1,h_2)}}
    \begin{pmatrix}
        a & 0\\
        0 & b
    \end{pmatrix}.
    \end{align}
    There are two constraints on $h_1$ and $h_2$, also $(h_1,h_2)\neq (0,0)$. Therefore, the number of possible transvections is $\left(\frac{N}{2}\right)^2-1 = \frac{N^2}{4}-1$.
    
    \item[(b)] Applying a transvection with Tr$(ah_2)=1$ and Tr$(bh_1)=1$, we obtain
    \begin{align}
    \begin{pmatrix}
        a & 0\\
        0 & b
    \end{pmatrix}\xrightarrow{Z_{(h_1,h_2)}}
    \begin{pmatrix}
        a+h_1 & h_2\\
        h_1 & b+h_2
    \end{pmatrix}
    \end{align}
    and the new orbit invariant is
    \begin{align*}
        (a+h_1)(b+h_2)+h_1h_2 = ab+ah_2+bh_1.
    \end{align*}
    Since $\text{Tr}(ab)+\text{Tr}(ah_2)+\text{Tr}(bh_1) = 1$, the resulting Pauli pair is not an edge.
    The products $ah_2$ and $bh_1$ range over all field elements with trace $1$. Given a field element $x$ with Tr$(x)=1$, the number of solutions to $x = ab+ ah_2 + bh_1$ is simply the number of solutions to $ah_2 + bh_1 =0$. There are $N/2$ transvections to each of the $N/2$ orbits.
    
    
    
    \item[(c)] Applying a transvection with Tr$(ah_2)=0$ and Tr$(bh_1)=1$, we obtain
    \begin{align}
    \begin{pmatrix}
        a & 0\\
        0 & b
    \end{pmatrix}\xrightarrow{Z_{(h_1,h_2)}}
    \begin{pmatrix}
        a & 0\\
        h_1 & b+h_2
    \end{pmatrix}
    \end{align}
    and the new orbit invariant is
    \begin{align*}
        a(b+h_2) = ab+ah_2.
    \end{align*}
    Since $\text{Tr}(ab)+\text{Tr}(ah_2) = 1$, the resulting Pauli pair is not an edge. A similar argument to that used in part (b) shows that there are $N/2$ transvections to each of the $N/2$ orbits. 
    
    \item[(d)] Applying a tranvection with Tr$(ah_2)=1$ and Tr$(bh_1)=0$, we obtain
    \begin{align}
    \begin{pmatrix}
        a & 0\\
        0 & b
    \end{pmatrix}\xrightarrow{Z_{(h_1,h_2)}}
    \begin{pmatrix}
        a+h_1 & h_2\\
        0 & b
    \end{pmatrix}
    \end{align}
    and the new orbit invariant is
    \begin{align*}
        (ah_1)b = ab+bh_1.
    \end{align*}
    The same argument used in part (c) shows that there are $N/2$ transvections to each of the $N/2$ orbits.
\end{enumerate}

There are $N^2-1$ transvections, and each is a symplectic matrix that preserves non-edges in the Pauli graph. Case (a) contributes to the diagonal component $\frac{N^2-4}{4(N^2-1)}I_{N/2}$ in $Q_1$ and cases (b), (c), and (d) contribute to $\frac{6N}{4(N^2-1)}J_{N/2}$ in $Q_1$. \qed
\end{proof}

\subsection{The Transvection Markov Chain on Edges}
\begin{theorem}
Consider the Markov process with state space consisting of all edge orbits. 
\edit{Set $M_1 = N-2$ and $M_2=\frac{N-2}{2}$}. Index the first \edit{$M_1$} rows and columns of the state transition matrix $Q_0$ by the type-$1$ orbits, and the remaining \edit{$M_2$} rows and columns by the type-$2$ orbits.
Then, $Q_0$ is given by
\edit{
\begin{align}\label{eq:transition-edge}
    Q_0=\frac{1}{4(N^2-1)}
    \begin{bmatrix}
    (N^2-4)I_{M_1} & NR^\mathsf{T}\\
    R & (N^2-4)I_{M_2} + 6NJ_{M_2},
    \end{bmatrix}
\end{align}
where $I_{M_1}$ and $I_{M_2}$ are the identity matrices, $J_{M_2}$ is the all ones matrix, $R$ is non-negative, each row sum of $R$ is $6N$, and each column sum of $R$ is $3N$.}
\end{theorem}

\begin{proof}
We determine the lower right block of the transition matrix $Q_0$ by making a slight modification to the proof of Theorem~\ref{thm:transition-non-edge}. Starting with Tr$(ab)=0$, we consider $x$ with \edit{Tr$(x)=0$} in cases (b), (c), and (d).
\edit{In each case, if $x=0$, we get a type-$1$ orbit and there are $N/2$ such transvections. Thus, the matrix $R/4$ in the lower left block has row sum $3N/2$. 
If $x \neq 0$, of which there are $(N-2)/2$ cases, we a get a transition to one of the $(N-2)/2$ type-$2$ orbits and there are still $N/2$ such transvections in each case. This contributes to the $J_{M_2}$ term in the lower right block. Finally, case (a) from Theorem~\ref{thm:transition-non-edge} produces the identity component.} 

We now start with a type-$1$ edge
$\begin{psmallmatrix}
a & 0\\
b & 0
\end{psmallmatrix}$, where $a\neq b$, $ab\neq 0$, and has orbit invariant \edit{$\frac{a}{b}$}, and apply a random transvection $Z_{(h_1,h_2)}$, where $h_1,h_2\in\mathbb{F}_{2^m}$ and $(h_1,h_2)\neq (0,0)$. 
By distinguishing four cases similar to the proof of Theorem~\ref{thm:transition-non-edge}, we notice that the upper left block of $Q_0$, which describes the probablity of transiting from type-$1$ orbits to type-$1$ orbits by transvections, contains only the diagonal component \edit{$\frac{N^2-4}{4(N^2-1)}I_{M_1}$}. 

Since transvections are self-inverse, the upper right block of $Q_0$ must be some scalar multiple of $R^\mathsf{T}$. The transvection Markov chain on edges is irreducible, so by the Perron-Frobenius Theorem,  there is a unique stationary distribution. Since the uniform distribution on edges is stationary, we observe that
\edit{
\begin{align}\label{eq:w_1}
    {\bf w_1} = \begin{bmatrix}
    \frac{1}{N} & \ldots & \frac{1}{N} & 1 & \ldots & 1
    \end{bmatrix} \in\mathbb{R}^{N-2+\frac{N-2}{2}}
\end{align}
is the stationary distribution of the Markov chain on edge orbits. Given that the row sum of $R$ is $6N$, the upper right block is $NR^\mathsf{T}$.} \qed
\end{proof}

\begin{example}
Here $m=3$, and we derive the matrix $R$.
Starting with a type-$2$ edge $\begin{psmallmatrix}
a & 0\\
0 & b 
\end{psmallmatrix}$ for which the orbit invariant is $ab=\alpha$, we consider the following three cases:

\begin{enumerate}
    \item[(a)] Applying a transvection with Tr$(ah_2)=1$ and Tr$(bh_1)=1$, we obtain
    \begin{align}
    \begin{pmatrix}
        a & 0\\
        0 & b
    \end{pmatrix}\xrightarrow{Z_{(h_1,h_2)}}
    \begin{pmatrix}
        a+h_1 & h_2\\
        h_1 & b+h_2
    \end{pmatrix},
    \end{align}
    with determinant
    \begin{align*}
        (a+h_1)(b+h_2)+h_1h_2 = ab+ah_1+bh_2.
    \end{align*}
    When $ah_1+bh_2=ab=\alpha$, the resulting Pauli pair is a type-$1$ edge with orbit invariant
    \edit{
    \begin{align*}
         \frac{a+h_1}{h_1}=\frac{ab+bh_1}{bh_1}=\frac{\alpha}{bh_1}+1.
    \end{align*}
    Since $bh_1$ takes the values $1$, $\alpha^3$, $\alpha^5$, or $\alpha^6$ with equal probability, the corresponding orbit invariants take the values
    \begin{align}
        \frac{\alpha}{1}+1=\alpha^3, \ \frac{\alpha}{\alpha^3}+1=\alpha^4, \ \frac{\alpha}{\alpha^5}+1=\alpha, \text{ or } \frac{\alpha}{\alpha^6}+1=\alpha^6 
    \end{align}}
    with equal probability.
    
    \item[(b)] Applying a transvection with Tr$(ah_2)=0$ and Tr$(bh_1)=1$, we obtain
    \begin{align}
    \begin{pmatrix}
        a & 0\\
        0 & b
    \end{pmatrix}\xrightarrow{Z_{(h_1,h_2)}}
    \begin{pmatrix}
        a & 0\\
        h_1 & b+h_2
    \end{pmatrix}
    \end{align}
    with determinant
    \begin{align*}
        a(b+h_2) = ab+ah_2.
    \end{align*}
    When $ah_2=ab=\alpha$, the resulting Pauli pair is a type-$1$ edge with orbit invarint
    \edit{
    \begin{align*}
        \frac{a}{h_1}=\frac{ab}{bh_1}=\frac{\alpha}{bh_1}.
    \end{align*}}
    
    Again, $bh_1$ takes the values $1$, $\alpha^3$, $\alpha^5$, or $\alpha^6$ with equal probability, so the corresponding orbit invariants take the values
    \edit{
    \begin{align}
        \frac{\alpha}{1} = \alpha, \ \frac{\alpha}{\alpha^3} = \alpha^5, \ \frac{\alpha}{\alpha^5} = \alpha^3 , \text{ or } \frac{\alpha}{\alpha^6} = \alpha^2
    \end{align}}
    with equal probability.
    
    \item[(c)] By symmetry, the result of applying a transvection with Tr$(ah_2)=1$ and Tr$(bh_1)=0$ is the same as part (b).
\end{enumerate}

A type-$2$ edge with orbit invariant $\alpha$ is equally likely to transition to any of the type-$1$ orbits. The same conclusion holds for a type-$2$ edge with orbit invariant $\alpha^2$ or $\alpha^4$, and so
\edit{\begin{align}
    R_3 = \begin{blockarray}{ccccccc}
& \alpha & \alpha^2 & \alpha^3 & \alpha^4 &\alpha^5 & \alpha^6 \\
\begin{block}{c[cccccc]}
  \alpha & 1 & 1 & 1 & 1 & 1 & 1\\
  \alpha^2 & 1 & 1 & 1 & 1 & 1 & 1\\
  \alpha^4 & 1 & 1 & 1 & 1 & 1 & 1\\
\end{block}
\end{blockarray}\times 8.
\end{align}}
\end{example}

\begin{remark}
When $m=3$, the matrix $R$ is a scalar multiple of a all ones matrix, but for general $m$, it may not be the case that a type-$2$ orbit is equally likely to transition to all type-$1$ orbits. For example, when $m=4$ and $\alpha$ is a root of $x^4+x+1$. \edit{Suppose we start with a type-2 edge $\begin{psmallmatrix}
a & 0\\
0 & b 
\end{psmallmatrix}$ for which the orbit invariant is $ab=\alpha$, then the number of possible transvections to achieve each type-1 orbits is
\begin{align}
\begin{blockarray}{ccccccccccccccc}
& \alpha & \alpha^2 & \alpha^3 & \alpha^4 &\alpha^5 & \alpha^6 & \alpha^7 & \alpha^8 & \alpha^9 & \alpha^{10} &\alpha^{11} & \alpha^{12} &\alpha^{13} & \alpha^{14} \\
\begin{block}{c[cccccccccccccc]}
\alpha & 1 & 2 & 1 & 1 & 3 & 2 & 2 & 2 & 2 & 3 & 1 & 1 & 2 & 1\\
\end{block}
\end{blockarray}
\end{align}}
\end{remark}

\subsection{Eigenvectors and Eigenvalues}\label{subsec:eigen}

We first discuss the eigenvectors and eigenvalues of $Q_1$. The all-one vector ${\bf x_1}$ of length $N/2$ is an eigenvector of $Q_1$ with eigenvalue $1$, i.e., ${\bf x_1}Q_1 = {\bf x_1}$. Any vector orthogonal to ${\bf x_1}$ is an eigenvector of $J_{N/2}$ with eigenvalue $0$. Then, for any vector ${\bf x_2}$ such that ${\bf x_2}J_{N/2} = 0$, we have
\begin{align}
    {\bf x_2}Q_1 = \frac{N^2-4}{4(N^2-1)}{\bf x_2}.
\end{align}
Therefore, the second largest eigenvalue of $Q_1$ is
\begin{align}
    \lambda_{Q_1} = \frac{N^2-4}{4(N^2-1)}
\end{align}
with multiplicity $\frac{N}{2}-1$.

We proceed to discuss the eigenvectors and eigenvalues of $Q_0$. Define
\edit{\begin{align}
    J'=\frac{1}{4(N^2-1)(N-2)}
    \begin{bmatrix}
    (N^2-4)J_{M_1} & 6N^2 J_{M_1 \times M_2}\\
    6NJ_{M_2\times M_1} & (8N^2-12N-8)J_{M_2}
    \end{bmatrix},
\end{align}
where $J_{M_1\times M_2}$ and $J_{M_2\times M_1}$ are the $M_1\times M_2$ and $M_2\times M_1$ all ones matrices.}
By direct calculation we obtain
\begin{align}
    J'Q_0=Q_0J'=J'J'.
\end{align}
Then, for any eigenvector ${\bf w}$ of $J'$ with eigenvalue $\lambda$, we have
\begin{align}
    {\bf w}J'Q_0={\bf w}J'J'
    \Rightarrow \lambda {\bf w} Q_0=\lambda^2{\bf w} \Rightarrow {\bf w} Q_0=\lambda{\bf w}.
\end{align}
Therefore, ${\bf w_1}$ in~\eqref{eq:w_1} 
and
\edit{
\begin{align}
    {\bf w_2} = \begin{bmatrix}
    1 & \ldots & 1 & -2 & \ldots & -2
    \end{bmatrix} \in\mathbb{R}^{N-2+\frac{N-2}{2}}
\end{align}}
are two (left) eigenvectors for both $J'$ and $Q_0$ with corresponding eigenvalues $1$ and $\frac{N^2-6N-4}{4(N^2-1)}$. The eigenvector ${\bf w_1}$ is the stationary distribution of $Q_0$.

Additionally, for any eigenvector ${\bf v}$ of $Q_0$ with eigenvalue $\lambda$, we have
\begin{align}
    \edit{{\bf v}J'J'}={\bf v}Q_0J'=\lambda{\bf v}J'.
\end{align}
Thus, \edit{either} ${\bf v}J'=0$ or ${\bf v}J'$ is an eigenvector of \edit{$J'$} with eigenvalue $\lambda$. It remains to calculate the eigenvalues associated with eigenvectors ${\bf v}$ of $Q_0$ that satisfy ${\bf v}J'=0$. The action of $Q_0$ on these eigenvectors is given by the matrix
\edit{\begin{align}
    Q_0'=\frac{1}{4(N^2-1)}
    \begin{bmatrix}
    (N^2-4)I_{M_1} & NR^\mathsf{T}\\
    R & (N^2-4)I_{M_2}
    \end{bmatrix}.
\end{align}}
Setting
\edit{
\begin{align}
    R'=\frac{1}{4(N^2-1)}
    \begin{bmatrix}
    0 & NR^\mathsf{T}\\
    R & 0
    \end{bmatrix},
\end{align}}
we have
\edit{
\begin{align}
    R'^2=\frac{1}{16(N^2-1)^2}
    \begin{bmatrix}
    NR^\mathsf{T}R & 0\\
    0 & NRR^\mathsf{T}
    \end{bmatrix}.
\end{align}
Since the row sum of $R$ is $6N$ and the column sum is $3N$, the largest singular value of $R$ is $3\sqrt{2}N$. Therefore, the largest eigenvalue of $R'^2$ is $\frac{N(3\sqrt{2}N)^2}{16(N^2-1)^2}$. Then, the largest eigenvalue of $R'$ is $\frac{3\sqrt{2}N\sqrt{N}}{4(N^2-1)}$. If $\lambda_{Q_0}$ is the second largest eigenvalue of $Q_0$, then
\begin{align}\label{eq:lambda_Q}
    \lambda_{Q_0} < \frac{N^2-4+3N\sqrt{2N}}{4(N^2-1)}
\end{align}
and the minimum eigenvalue $\lambda_{\min, Q_0}$ of $Q_0$ is positive when $m\geq 5$ (i.e., $N\geq 32$) since
\begin{align}
    \lambda_{\min, Q_0} > \frac{N^2-4 - 3N\sqrt{2N}}{4(N^2-1)} > 0.
\end{align}}

\section{Convergence Analysis}\label{sec:convergence}
In this section, we prove that the transvection Markov process gives an $\epsilon$-approximate unitary $3$-design and analyze its convergence rate using the second largest eigenvalues obtained in Section~\ref{subsec:eigen}. The proof also receives inspirations from Section 6 in \cite{Harrow-cmp09} and how Webb~\cite{Webb-arxiv16} proved that the Pauli $2$-mixing forms an exact unitary $3$-design. 

\subsection{The Approximate Unitary $k$-Designs}

\begin{definition}
[{\hspace{1sp}\cite[Chapter 9.1.6]{Wilde_2017_book}, \cite[Definition 2.4]{Harrow-cmp09}}]
For quantum channels $\mathcal{N}, \mathcal{M}\in C(\mathcal{X})$, the \emph{diamond-norm distance} is given by
\begin{align}
    \norm{\mathcal{N} - \mathcal{M}}_\diamond = \sup_{\rho} \norm{(\mathcal{N} \otimes I_R)(\rho) - (\mathcal{M} \otimes I_R)(\rho)}_1
\end{align}
where $\text{dim}(R) = \text{dim}(\mathcal{X})$. The diamond norm tells us the distinguishability of two quantum channels in an operational sense. 
\end{definition}

\begin{definition}[{\hspace{1sp}\cite[Definition 2.5]{Harrow-cmp09}}]
$\mathcal{G}_\mathcal{E}$ is an \emph{$\epsilon$-approximate unitary $k$-design} if 
\begin{align}
    \norm{\mathcal{G}_\mathcal{E} - \mathcal{G}_H}_\diamond \leq \epsilon.
\end{align}
\end{definition}

\begin{theorem}\label{thm:sample_transvections}
Random sampling of \edit{$O(\log(N^5/\epsilon))$} transvections, followed by a random element from $\mathfrak{P}_{m} \cong \text{PSL}(2,2^m)$ and a random Pauli matrix $D(a, b)$ gives a $\epsilon$-approximate unitary $3$-design.
\end{theorem}

We will use the rest of Section~\ref{sec:convergence} to prove Theorem~\ref{thm:sample_transvections}. 
Denote $\overline{P}_N$ as the set of all $E(a, b)$ Pauli matrices. Since $\overline{P}_N$ forms an orthogonal basis for $L(\mathbb{C}^N)$, each linear operator $\rho \in L((\mathbb{C}^N)^{\otimes 6})$ can be written as
\begin{align}
    \rho = \frac{1}{N^3} \sum_{\substack{p_i\in\overline{P}_N, \\ i\in \{1, \ldots, 6\}}} \Gamma(p_1, \ldots, p_6)p_1 \otimes \cdots \otimes p_6,
\end{align}
where
\begin{align}
    \Gamma(p_1, \ldots, p_6) = \frac{1}{N^3} \text{Tr}((p_1 \otimes \cdots \otimes p_6) \rho) \in \mathbb{R}
\end{align}
and
\begin{align}
    \sum_{\substack{p_i\in\overline{P}_N, \\ i\in \{1, \ldots, 6\}}} \Gamma^2(p_1, \ldots, p_6) \leq 1.
\end{align}
Because of this nice expansion into Pauli elements, we only need to understand the effect of $\mathcal{G}_\mathcal{E}$ and $\mathcal{G}_H$ acting on each element of $\overline{P}_N^{\otimes 3}$. Notice that Theorem~\ref{thm:sample_transvections} gives an exact unitary $2$-design since it still acts transitively on vertices. Then, it follows from Webb~\cite[Proof of Lemma 4]{Webb-arxiv16} that when $\rho = p_1 \otimes p_2 \otimes I$ or $\rho = p_1 \otimes p_2 \otimes p_3$ with $p_1p_2p_3 \not\propto I$ (up to permutations of the underlying registers), we have
\begin{align}
    \mathcal{G}_\mathcal{E}(\rho) = \mathcal{G}_H (\rho).
\end{align}
Therefore, we only need to consider the case when $\rho = p_1 \otimes p_2 \otimes p_1p_2$ and $p_1 \neq p_2$: 
\begin{align} \label{eq:appro_proof_mid}
 \norm{\mathcal{G}_\mathcal{E} - \mathcal{G}_H}^2_\diamond = & \sup_{\rho} \norm{(\mathcal{G}_\mathcal{E} \otimes I_R)(\rho) - (\mathcal{G}_H \otimes I_R)(\rho)}^2_1 \nonumber\\
\leq & \edit{N^6} \sup_\rho \norm{(\mathcal{G}_\mathcal{E} \otimes I_R)(\rho) - (\mathcal{G}_H \otimes I_R)(\rho)}_2^2 \nonumber \\
= & \sup_\rho \Biggr\| \sum_{p_i\in\overline{P}_N} \Gamma(p_1, \ldots, p_6) \bigg[ \mathcal{G}_\mathcal{E} (p_1\otimes p_2 \otimes p_3) \otimes p_4 \otimes p_5 \otimes p_6  \nonumber\\
& \hspace{1.25cm} - \mathcal{G}_H (p_1\otimes p_2 \otimes p_3) \otimes p_4 \otimes p_5 \otimes p_6 \bigg] \Biggr\|_2^2 \nonumber \\
= & \sup_\rho \Biggr\| \sum_{\substack{p_i\in\overline{P}_N  \\ p_3 = p_1p_2, p_1 \neq p_2}} \Gamma(p_1, \ldots, p_6) \bigg[ \mathcal{G}_\mathcal{E} (p_1\otimes p_2 \otimes p_1p_2) 
\nonumber\\
& \hspace{1.25cm} - \mathcal{G}_H (p_1\otimes p_2 \otimes p_1p_2) \bigg] \otimes p_4 \otimes p_5 \otimes p_6 \Biggr\|_2^2.
\end{align}
Define
\begin{align}
    \mathcal{E}_{(p_1, p_2) \to (q_1, q_2)} = \{(\alpha, U)\in\mathcal{E} \colon Up_1U^\dagger = q_1, Up_2U^\dagger = q_2\}.
\end{align}
Suppose the transition matrices for the Markov chains of sampling random transvections on all edges and non-edges are $P_0$ and $P_1$ respectively. Then, if we start from the Pauli pair $(p_1, p_2)$, the probability of reaching $(q_1, q_2)$ after $t$ steps on the chain is
\begin{align}
    g_t^{[p_1, p_2]} (q_1, q_2; p_1, p_2) = {\bf e}_{(p_1, p_2)} P_{[p_1, p_2]}^t {\bf e}^\mathsf{T}_{(q_1, q_2)},
\end{align}
where $[p_1, p_2] = 0$ if $p_1$ and $p_2$ commute, otherwise $[p_1, p_2] = 1$ and ${\bf e}_{(p_1, p_2)}$ is a unit row vector where only position $(p_1, p_2)$ is non-zero. Then, for the ensemble $\mathcal{E}$ generated after $t$ steps on the corresponding Markov chain, we have
\begin{align}
\mathcal{G}_\mathcal{E} (p_1\otimes p_2 \otimes p_1p_2) =& \sum_{(\alpha, U)\in\mathcal{E}} \alpha U p_1 U^\dagger \otimes U p_2 U^\dagger \otimes U p_1 U^\dagger Up_2 U^\dagger \nonumber\\
= & \sum_{\substack{q_1, q_2 \in\overline{P}_N \colon \\ [q_1, q_2] = [p_1, p_2]}} \sum_{(\alpha, U)\in \mathcal{E}_{(p_1, p_2)\to(q_1, q_2)}} \alpha q_1 \otimes q_2 \otimes q_1q_2 \nonumber\\
= & \sum_{\substack{q_1, q_2 \in\overline{P}_N \colon \\ [q_1, q_2] = [p_1, p_2]}} g_t^{[p_1, p_2]} (q_1, q_2; p_1, p_2) q_1 \otimes q_2 \otimes q_1q_2
\end{align}
and
\begin{align}
\mathcal{G}_H (p_1\otimes p_2 \otimes p_1p_2) = & \sum_{\substack{q_1, q_2 \in\overline{P}_N \colon \\ [q_1, q_2] = [p_1, p_2]}} g_\infty^{[p_1, p_2]} (q_1, q_2; p_1, p_2) q_1 \otimes q_2 \otimes q_1q_2.
\end{align}
Then, by continuing~\eqref{eq:appro_proof_mid} and following the orthogonality of the Pauli operators under the Hilbert-Schmidt inner product, we obtain
\begin{align}\label{eq:appro_proof_last}
\norm{ \mathcal{G}_\mathcal{E} - \mathcal{G}_H }^2_\diamond 
= & \sup_\rho \Biggr\| \sum_{\substack{p_i\in\overline{P}_N \\ [p_1, p_2] = 0}} \Gamma(p_1,\ldots, p_6)\sum_{\substack{q_1, q_2 \in\overline{P}_N \colon \\ [q_1, q_2] = 0}} \bigg[ g_t^0 (q_1, q_2; p_1, p_2)  \nonumber\\
& \hspace{1cm} - g_\infty^0 (q_1, q_2; p_1, p_2) \bigg] q_1\otimes q_2 \otimes q_1q_2 \otimes p_4 \otimes p_5 \otimes p_6 \nonumber\\
& + \sum_{\substack{p_i\in\overline{P}_N \\ [p_1, p_2] = 1}} \Gamma(p_1,\ldots, p_6)\sum_{\substack{q_1, q_2 \in\overline{P}_N \colon \\ [q_1, q_2] = 1}} \bigg[ g_t^1 (q_1, q_2; p_1, p_2)  \nonumber\\
& \hspace{1cm} - g_\infty^1 (q_1, q_2; p_1, p_2) \bigg] q_1\otimes q_2 \otimes q_1q_2 \otimes p_4 \otimes p_5 \otimes p_6 \Biggr\|_2^2 \nonumber\\
= & \edit{N^6}\sup_\rho  \biggr[ \sum_{ [p_1, p_2] = 0} \Gamma^2(\cdots)\sum_{[q_1, q_2] = 0} (g_t^0 (\cdots) - g_\infty^0 (\cdots) )^2  \nonumber\\
& \hspace{1cm} + \sum_{ [p_1, p_2] = 1} \Gamma^2(\cdots)\sum_{[q_1, q_2] = 1} (g_t^1 (\cdots) - g_\infty^1 (\cdots) )^2 \biggr].
\end{align}
\edit{We shall show in Section \ref{sec:orbits_mixing_time} that when $t=O(\log(N^5/\epsilon))$, we have 
\begin{align} \label{eq:g_t_infty^0}
    \sum_{[q_1, q_2] = 0} (g_t^0 (\cdots) - g_\infty^0 (\cdots) )^2 < \left(\frac{\epsilon}{N^3}\right)^2
\end{align}
and 
\begin{align}\label{eq:g_t_infty^1}
    \sum_{[q_1, q_2] = 1} (g_t^1 (\cdots) - g_\infty^1 (\cdots) )^2 < \left(\frac{\epsilon}{N^3}\right)^2
\end{align}
using the second largest eigenvalues of $Q_0$ and $Q_1$ defined in~\eqref{eq:transition-edge} and~\eqref{eq:transition-non-edge}. 
Therefore, we can prove Theorem~\ref{thm:sample_transvections} by further simplifying~\eqref{eq:appro_proof_last} into
\begin{align}
\Vert \mathcal{G}_\mathcal{E} - \mathcal{G}_H\Vert^2_\diamond \leq N^6\sup_\rho \sum \Gamma^2(\cdots) \left(\frac{\epsilon}{N^3}\right)^2 \leq \epsilon^2.
\end{align}}

\subsection{Orbits Mixing Time}\label{sec:orbits_mixing_time}

We first find the connection between the probability distribution of the Markov chain on edges and that of edge orbits. Suppose 
\edit{\begin{align}\label{eq:general_x}
    {\bf x} = [x_1, \ldots, x_{N-2}, x_{N-1}, \ldots, x_{\frac{3}{2}(N-2)}]
\end{align}}
is a probability distribution on all edge orbits. Sampling Clifford elements will distribute the probability mass equally among all the edges within each orbit. Since each type-$1$ orbit has $N^2-1$ elements and \edit{each type-$2$ orbit has $(N^2-1)N$ elements}, we define the corresponding probability distribution on all edges after sampling a random Clifford element as
\edit{\begin{align}
    f({\bf x}) = \left[\frac{x_1}{N^2-1}, \ldots, \frac{x_{N/2}}{(N^2-1)N}, \ldots\right],
\end{align}
where each entry $\frac{x_i}{N^2 - 1}$ is repeated $N^2-1$ times, for $1 \leq i \leq (N-2)$, and each entry $\frac{x_i}{(N^2-1)N}$ is repeated $(N^2-1)N$ times, for $i \geq N-1$.
Then, we see that
\begin{align}
\Vert f({\bf x}) - f({\bf y})\Vert_1 =& \sum_{i=1}^{(N^2-1)(N^2-4)/2} |f({\bf x})_i - f({\bf y})_i| \nonumber \\
=& \sum_{i=1}^{N-2} (N^2-1) \left|\frac{x_i}{N^2-1} - \frac{y_i}{N^2-1}\right| \nonumber \\
& \hspace{0.5cm} + \sum_{i=N-1}^{\frac{3}{2}(N-2)} (N^2-1)N \left|\frac{x_i}{(N^2-1)N} - \frac{y_i}{(N^2-1)N}\right| \nonumber\\
=& \sum_{i=1}^{\frac{3}{2}(N-2)} |x_i - y_i| = \Vert {\bf x} - {\bf y}\Vert_1.
\end{align}}

\begin{definition}Let $Q$ be the transition matrix of an irreducible and aperiodic Markov process whose stationary distribution is $\pi$. The mixing time $\tau$ is given by
\begin{align}
    \tau(\epsilon) = \max_{{\bf s}\colon \Vert {\bf s}\Vert_1 = 1} \min_{t} \{t\geq 0 \colon \Vert {\bf s}Q^t - \pi\Vert_1 \leq \epsilon\}.
\end{align}
\end{definition}

\begin{theorem}[{\hspace{1sp}\cite[Theorem 4.5]{Harrow-cmp09}}]
The mixing time can be bounded above as
\begin{align}
    \tau(\epsilon) \leq \frac{1}{\Delta} \ln{\frac{1}{\pi_* \epsilon}},
\end{align}
where $\pi_* = \min \pi(x)$ and $\Delta = \min(1-\lambda_2, 1+\lambda_{\min} )$.
Here, $\lambda_2$ is the second largest eigenvalue and $\lambda_{\min}$ is the smallest. 
\end{theorem}

Therefore, following from~\eqref{eq:w_1} and~\eqref{eq:lambda_Q}, the mixing time for $Q_0$ is bounded by
\edit{\begin{align}
    t_{\min} = \tau\left(\frac{\epsilon}{N^3}\right) \leq & \frac{1}{1-\lambda_{Q_0}}\ln{\frac{1}{\min\left(\frac{{\bf w}_1}{\Vert {\bf w}_1 \Vert_1}\right) \frac{\epsilon}{N^3}}} \nonumber\\
    \leq & \frac{1}{1-\frac{N^2-4+3N\sqrt{2N}}{4(N^2-1)}}\ln{\frac{1}{ \frac{\epsilon}{N^3} (\frac{1}{N})/{(\frac{N^2-4}{2N})} }} \nonumber\\
    \approx & \frac{1}{1-1/4} \ln{\frac{N^3(N^2-4)}{2\epsilon}} = O(\log (N^5/\epsilon))
\end{align}\cite{Singal-21}}
and $t_{\min}$ is larger than the mixing time of $Q_1$ defined in~\eqref{eq:transition-non-edge}. Since
\begin{align}
    \Vert f({\bf x}) - f({\bf y})\Vert_2 \leq \Vert f({\bf x}) - f({\bf y})\Vert_1 = \Vert {\bf x} - {\bf y}\Vert_1, 
\end{align}
both~\eqref{eq:g_t_infty^0} and~\eqref{eq:g_t_infty^1} can be upper bounded by \edit{$\left(\frac{\epsilon}{N^3}\right)^2$}. \edit{Theorem \ref{thm:sample_transvections} is thus proved. }


\section{Conclusion}\label{sec:conclusion}

In this paper, we have proved that we can obtain an $\epsilon$-approximate unitary $3$-design by random sampling \edit{$O(\log(N^5/\epsilon))$} transvections followed by a random Clifford $\text{PSL}(2, 2^m)$ element and a random Pauli matrix. 

The key to an exact unitary $3$-design is the transitivity on the Pauli element pairs of the same commutativity. The transvection Markov process that we designed exactly enables all Pauli pairs to converge to such a uniform stationary distribution. We first start from a well-established unitary $2$-design, a symplectic subgroup isomorphic to $\text{PSL}(2, 2^m)$ based on Kerdock sets. This unitary $2$-design acts transitively on Pauli elements, but partitions the Pauli pairs into multiple orbits. Using a finite field representation of Paulis, and the fact that transvections generate the symplectic group, we characterize the orbits of Pauli pairs and analyze how each transvection acts on the orbits. 
Finally, we analyze the convergence rate to an $\epsilon$-approximate $3$-design using the second largest eigenvalues of the transition matrices of the ``edge''- and ``orbit''-Markov chains. 

Interestingly, when we look at the history of related research, scientists have already discovered that the Clifford group forms a unitary $2$-design early in 2001~\cite{DiVincenzo-it02}. But it was not proved until 2016 that it is also a unitary $3$-design and it is the minimal $3$-design except for dimension $4$~\cite{Zhu-aps17}~\cite{Webb-arxiv16}. Our work fills in the gap by constructing the $3$-design using the exact building blocks provided by the $2$-design. 

Given that M{\o}lmer-S{\o}rensen gates (i.e., $2$-qubit transvections~\cite{Rengaswamy-phd20}) are native gates in trapped-ion systems, it would be interesting to investigate if employing only these transvections in our Markov chain suffices to produce an approximate unitary $3$-design.
Moreover, it might be possible to extend this Markov chain approximation framework to create unitary $k$-designs for $t > 3$ by mixing the Kerdock design with, perhaps, a generalization of symplectic transvections.

\section*{Acknowledgments}
\edit{We are grateful to Tanmay Singal and Min-Hsiu Hsieh for pointing out that we needed to define the Pauli graph as a directed graph in order to properly analyze the Transvection Markov Chain. They also communicated that it is possible to improve the convergence analysis to $O(\log (N^3/\epsilon))$ \cite{Singal-21}.}

The work of N. Rengaswamy and R. Calderbank was supported in part by the National Science Foundation (NSF) under Grant no. 1908730.

\bibliographystyle{IEEEtran}
\bibliography{WCLabrv,WCLnewbib}

\end{document}